\documentclass[nonacm]{acmart}

\newif\iflong
\newif\ifshort

\longtrue
    
\iflong
\else
\shorttrue
\fi

\usepackage{soul}
\usepackage{url}
\usepackage[utf8]{inputenc}
\usepackage{graphicx}
\usepackage{amsmath,amsthm,amsfonts}%
\usepackage{mathtools}
\usepackage{booktabs}
\usepackage{algorithm}
\usepackage{algorithmic}
\usepackage[switch]{lineno}
\usepackage{boxedminipage}
\usepackage{enumerate}
\usepackage{comment}

\usepackage{todonotes}

\usepackage{cleveref}

\newtheorem{theorem}{Theorem}
\newtheorem{lemma}[theorem]{Lemma}
\newtheorem{corollary}[theorem]{Corollary}

\newcommand{\cC}{\mathcal{C}}
\newcommand{\cF}{\mathcal{F}}
\newcommand{\cS}{\mathcal{S}}
\newcommand{\bigoh}{\mathcal{O}}
\newcommand{\cQ}{\mathcal{Q}}
\newcommand{\NN}{\mathbb{N}}
\newcommand{\XP}{XP}

\newcommand{\QBF}{QBFSAT}
\newcommand{\dQBF}{$d$-QBFSAT}
\newcommand{\FEdQBF}[1]{$\forall\exists#1$-QBFSAT}
\newcommand{\FEQBF}{$\forall\exists\text{QBFSAT}$}

\DeclareMathOperator{\ar}{ar}

\DeclarePairedDelimiter\ceil{\lceil}{\rceil}

\newcommand{\pbDef}[3]{%
  \noindent
  \begin{center}
  \begin{boxedminipage}{0.49\textwidth}
  {\sc #1}\\[5pt]
  \begin{tabular}{l p{0.7\textwidth}}
  {\sc Instance}: & #2\\
  {\sc Question}: & #3
  \end{tabular}
  \end{boxedminipage}
  \end{center}
}

\newcommand{\CGIS}{\textsc{Clause-Graph Independent Set}}
\newcommand{\TAUT}{\textsc{Or-CNF Tautology}}

\title{Solving Quantified Boolean Formulas with Few Existential Variables}

\author{Leif Eriksson}
\email{leif.eriksson@liu.se}
\affiliation{%
  \institution{Department of Computer and Information Science, Linköping University}
  \city{Linköping}
  \country{Sweden}
}

\author{Victor Lagerkvist}
\email{victor.lagerkvist@liu.se}
\affiliation{%
  \institution{Department of Computer and Information Science, Linköping University}
  \city{Linköping}
  \country{Sweden}
}

\author{George Osipov}
\email{george.osipov@liu.se}
\affiliation{%
  \institution{Department of Computer and Information Science, Linköping University}
  \city{Linköping}
  \country{Sweden}
}

\author{Sebastian Ordyniak}
\email{sordyniak@gmail.com}
\affiliation{%
  \institution{School of Computing, University of Leeds}
  \city{Leeds}
  \country{UK}
}

\author{Fahad Panolan}
\email{fahad.panolan@gmail.com}
\affiliation{%
  \institution{School of Computing, University of Leeds}
  \city{Leeds}
  \country{UK}
}

\author{Mateusz Rychlicki}
\email{mkrychlicki@gmail.com}
\affiliation{%
  \institution{School of Computing, University of Leeds}
  \city{Leeds}
  \country{UK}
}

\begin{document}

\begin{abstract}
    The {\em quantified Boolean formula} (QBF) problem is an important
    decision problem generally viewed as the archetype for PSPACE-completeness.
  Many problems of central interest in AI are in general not included in NP, e.g.,  planning, model checking, and non-monotonic reasoning, and for such problems QBF has successfully been used as a modelling tool.
  However, solvers for QBF are not as advanced as state of the art SAT solvers, which has prevented QBF from becoming a universal modelling language for PSPACE-complete problems. A theoretical explanation is that QBF (as well as many other PSPACE-complete problems) lacks natural {\em parameters} guaranteeing {\em fixed-parameter tractability} (FPT). 
  
  In this paper we tackle this problem and consider a simple but overlooked parameter: the number of existentially quantified variables. 
  This natural parameter is virtually unexplored in the literature which one might find surprising given the general scarcity of FPT algorithms for QBF.
  Via this parameterization we then develop a novel FPT algorithm applicable to QBF instances in conjunctive normal form (CNF) of bounded clause length. We complement this by a W[1]-hardness result for QBF in CNF of unbounded clause length as well as sharper lower bounds for the bounded arity case under the (strong) {\em exponential-time hypothesis}. 
  \end{abstract}

\maketitle

\section{Introduction}

The {\em quantified Boolean formula} (QBF) problem is the decision problem of verifying a formula 
$Q_1 x_1 \ldots Q_n x_n . \varphi(x_1, \ldots, x_n)$, where $\varphi(x_1, \ldots, x_n)$ is a propositional formula and $Q_i \in \{\forall, \exists\}$ for each $i$.
Throughout, we write \QBF{} (respectively QBF-DNF) for the subproblem restricted to formulas in conjunctive normal form (respectively disjunctive normal form), and \dQBF{} for the problem where clauses have maximum size $d \geq 1$. From a theoretical angle, the QBF problem serves as a foundational example of PSPACE-completeness, and by restricting the quantifier alternations, we get examples of complete problems for any class in the polynomial hierarchy. From a more practical point of view, the field of applied QBF solving has developed on the shoulders of SAT solving, which has seen tremendous advances in the last decade~\cite{DBLP:journals/cacm/FichteBHS23}. Arguably, the raison d'etre behind SAT solving is not only to solve a specific NP-complete problem faster than exhaustive search, but to provide a combinatorial framework applicable to any problem which can be reduced to SAT. This naturally includes any problem in NP but if one considers stronger reductions than Karp reductions (e.g., by allowing a superpolynomial running time or by viewing the SAT solver as an oracle)  more problems fall under the umbrella of SAT. However, this approach is not always optimal, and for many problems of importance in artificial intelligence, e.g., planning, model checking, and non-monotonic reasoning~\cite{qbfsurvey}, this approach is not ideal since the best reductions in the literature incur an exponential overhead. These problems are instead more naturally formulated via QBF. However, this comes with the downside that QBF solvers, despite steady advances, are not nearly as advanced as their SAT brethren. For more information about applied QBF solving we refer the reader to the handbook by Biere et al.~\cite{biere2021handbook}. 

To bridge the gap between SAT and QBF solving we need algorithmic breakthroughs for the latter. In this paper we analyze QBF from a {\em theoretical} perspective and are therefore interested in obtaining unconditionally improved algorithms. 
To analyze the complexity of QBF we use the influential paradigm of
{\em parameterized complexity} where the goal is to identify
structural properties of instances, represented by natural numbers
called {\em parameters}, such that one can effectively solve instances
with bounded parameter size. More formally,  for every instance $I$
of a computational problem we associate a parameter $k \in
\mathbb{N}$ with it and  then we are primarily interested in algorithms with a
running time bounded by $f(k) \cdot \mathrm{poly}(|I|)$ for a
computable function $f \colon \mathbb{N} \to \mathbb{N}$ depends on $k$ and a
polynomial function $\mathrm{poly}$ depends on the length $|I|$ of the
input $I$. Such algorithms are said to be {\em fixed-parameter
  tractable} (FPT). Thus, while $f$ is generally going to be
superpolynomial, an FPT algorithm may still be very competitive in
practice if the parameter is sufficiently small. For problems in NP
there exists a wealth of results~\cite{parameterizedalgorithmsbook},
but for PSPACE-complete problems the landscape is rather scarce in
comparison since there are fewer natural parameters to choose
from. For example, the go-to parameter for NP problems is {\em
  tree-width}, which measures how close a graph is to being a tree,
and which is typically sufficient to produce an FPT algorithm. But
this fails for QBF where it is even known that QBF is PSPACE-complete
for constant primal tree-width~\cite{ATSERIAS20141415}.
Here, the situation
becomes more manageable if one simultaneously bounds e.g.\ the number
of quantifier alternations~\cite{chen2004} but for the general QBF
problem the few FPT results that exists are primarily with respect to
more exotic parameters such as prefix pathwidth~\cite{EIBEN20201} and
respectful treewidth~\cite{ATSERIAS20141415} which takes the ordering
of the quantifiers into account. Two interesting counter examples are
(1) the FPT algorithm parameterized by primal vertex cover number 
by~\cite{EIBEN20201}, further optimized and simplified 
by~\cite{lampis_et_al:LIPIcs.IPEC.2017.26}, and (2) the
backdoor approach in~\cite{Samer2009} which
generalizes the classical tractable fragments of QBF in an FPT
setting. Moreover, very recently, tractability has been shown for two
parameters that fall between vertex cover number
and treewidth~\cite{DBLP:conf/lics/FichteGHSO23}.
Thus, while FPT results for QBF and related problems exist, they are comparably few in number and generally defined with respect to more complicated structural properties. Do simple parameters not exist, or have we been investigating the wrong ones? 

In this paper we demonstrate that a natural (and previously overlooked) parameter {\em does} exist: the number of existentially quantified variables. Thus, we bound the number of existentially quantified variables but otherwise make no restrictions on the prefix. While quantifier elimination techniques have been an important tool ever since the early days of QBF solving~\cite{ayari2002qubos,biere2004resolve} (see also Janota \& Marques-Silva~\cite{JANOTA201525} for a more recent discussion and a comparison to $Q$-resolution) the predominant focus has been to expand universal quantifiers since the removal of universal quantifiers produces instances that can be solved with classical SAT techniques (however, methods for expanding existential quantifiers have also been tried in practice~\cite{expansionexample}).
This is, for example, made explicit by  
Szeider \& de Haan~\cite{10.1007/978-3-319-09284-3_8} who prove that QBF-DNF parameterized by the number of universally quantified variables is FPT-reducible to SAT (parameterized by the number of all variables). However, they only prove para-NP-completeness of the problem which in the world of parameterized complexity is a far cry away from FPT. Conversely, we show that concentrating on existential variable elimination is much more lucrative since  in this case one can construct an FPT algorithm once all existential variables have been removed. Let us also remark that if a \QBF{} instance contains $k$ existential variables and $n - k$ universal variables, then the combination $n$ is not a relevant parameter: it makes the problem technically FPT but there are very few applications where one would expect the total number of variables to be bounded. Moreover, while \QBF{} for arbitrary but constant quantifier depth admits a moderately improved algorithm with a running time of the form $2^{n - n^{\Omega(1)}}$~\cite{DBLP:journals/eccc/SanthanamW13}, not even 3-QBFSAT restricted to universal followed by existential quantifiers admits an exponentially improved $2^{\varepsilon n} \cdot \mathrm{poly}(|\phi|)$ algorithm for $\varepsilon < 1$ under the so-called {\em strong exponential-time hypothesis}~\cite{CalabroIP13}.

After having defined the basic notions (Section~\ref{sec:preliminaries}), we obtain our major results in Section~\ref{sec:algorithms} as follows. First, given an instance $Q_1 x_1 \ldots \exists x_i \ldots Q_n x_n \phi(x_1, \ldots, x_n)$ we can remove the existential quantifier for $x_i$ by creating a disjunction of the two subinstances (with fresh variables) obtained by fixing $x_i$ to 0 or 1.
This extends to a quantifier elimination preprocessing scheme which reduces to the problem of checking whether a disjunction of $k$ formulas in $d$-CNF is a tautology or not. We call this problem \TAUT{} and then construct an FPT algorithm by converting it to the problem of finding an independent set in a certain well-structured graph (\CGIS{}). The latter problem has a strong combinatorial flair and we manage to construct a kernel with at most $kd!((k-1)d+1)^d$ vertices via the {\em sunflower lemma} of Erd{\"o}s \& Rado~\cite{Erdos60}. Put together, this results in an FPT algorithm with a running time of $2^{\bigoh(k2^k)} \cdot |\phi|$ for \QBF{} with constant clause size when parameterized by the number of existentially quantified variables.
At a first glance, this might not look terribly impressive compared to the naive $2^{n} \cdot |\phi|$ algorithm obtained by branching on $n - k$ universal and $k$ existential quantifiers, but we stress that the FPT algorithm is suitable for applications where the total number of existentially quantified variables is kept relatively small. Under this constraint our algorithm shows that one can effectively ignore the cost of universally quantified variables and solve the instance in polynomial time. Via an FPT reduction we also demonstrate (in Section~\ref{sec:qcsp}) that our FPT result straightforwardly can be extended to the more general problem {\em quantified constraint satisfaction} where variables take values in arbitrary finite domains. 

In Section~\ref{sec:lower_bounds} we show that the clause size dependency in our FPT algorithm is necessary under the conjecture that FPT $\neq$ W[1], which is widely believed conjecture in parameterized complexity. 
Specifically we show that \QBF{}, when parameterized by the number of existentially quantified variables (but not the arity), is W[1]-hard by reducing from the \textsc{Multipartite Independent Set} problem. 
Moreover, under the Exponential Time Hypothesis (ETH), the same reduction rules out algorithms
with running time $f(k) \cdot 2^{o(2^k)}$ for every computable
function $f : \NN \to \NN$.
We then proceed to establish a sharper lower bound under the (Strong) ETH for the specific case of clause size 3.
First, we have an easy $2^{o(k)}$ lower bound for 3-QBFSAT (under the exponential-time hypothesis) since 3-SAT can be viewed as a special case of 3-QBFSAT. However, under the strong exponential-time hypothesis we prove a markedly stronger bound: there is {\em no} constant $c$ such that 3-QBFSAT (even for only two quantifier blocks) is solvable in $c^k$ time, i.e. the problem is not solvable in $2^{\bigoh(k)}$ time. Thus, while our $2^{\bigoh(k2^k)}$ algorithm likely can be improved to some extent, we should not hope to obtain a $2^{\bigoh(k)}$ algorithm. This proof is based on an interesting observation: any instance with comparably few number of universally quantified variables can be solved by enumerating all possible assignments to the universal variables and then solving the remaining part with a 3-SAT algorithm. Thus, as also remarked by Calabro et al.~\cite{CalabroIP13}, instances with few existential variables are in a certain sense harder, which makes our FPT algorithm in Section~\ref{sec:algorithms} all the more surprising.

We close the paper with a discussion in Section~\ref{sec:discussion}. Most importantly, our FPT algorithm shows tractability of new classes of previously hard QBFs, and it would be interesting to investigate whether similar parameters could be used to study other hard problems outside NP. A promising candidate is the NEXPTIME-complete problem obtained by extending Boolean formulas with {\em Henkin quantifiers}, resulting in the {\em depedency QBF} (DQBF) formalism.

\section{Preliminaries} \label{sec:preliminaries}

In this section we briefly cover 
the necessary background
on parameterized complexity
and quantified Boolean formulas.
We assume familiarity with the basics of 
graph theory, cf.~\cite{DiestelBook}.
We use the notation $[n]$ for the set $\{1,\dots,n\}$ for every $n \in \NN$.

\paragraph{Computational Complexity.}

We follow~\cite{downey2013fundamentals}~and~\cite{fomin2019kernelization} 
in our presentation.
Let $\Sigma$ be a finite alphabet.
A parameterized problem $L$ is a subset of $\Sigma^* \times \NN$.
The problem $L$ is \emph{fixed-parameter tractable} (or, in FPT)
if there is an algorithm deciding 
whether an instance $(I, k) \in \Sigma^* \times \NN$ is in $L$
in time $f(k) \cdot |I|^c$, where
$f$ is some computable function and
$c$ is a constant independent of $(I, k)$.

Let $L, L' \subseteq \Sigma^* \times \NN$ be two parameterized problems.
A mapping $P : \Sigma^* \times \NN \to \Sigma^* \times \NN$ is a 
\emph{fixed-parameter (FPT) reduction from $L$ to $L'$}
if there exist computable functions $f,p : \NN \to \NN$
and a constant $c$ such that the following conditions hold:
\begin{itemize}
  \item $(I,k) \in L$ if and only if $P(I,k) = (I',k') \in L'$,
  \item $k' \leq p(k)$, and
  \item $P(I,k)$ can be computed in $f(k) \cdot |I|^c$ time.
\end{itemize}
Let $L \subseteq \Sigma^* \times \NN$ be a parameterized problem.
A \emph{kernelization (algorithm)} for $L$
is an algorithm that takes $(I, k) \in \Sigma^* \times \NN$ as input and
in time polynomial in $|(I,k)|$, 
outputs $(I',k') \in \Sigma^* \times \NN$ such that: \ifshort $(I,k)
  \in L$ if and only if $(I',k')\in L$, and $|I'|,k' \leq h(k)$ for
  some computable function $h$. \fi
\iflong
\begin{itemize}
  \item $(I,k) \in L$ if and only if $(I',k')\in L$, and 
  \item $|I'|,k' \leq h(k)$ for some computable function $h$.
\end{itemize}
\fi
The output $(I',k')$ of the kernelization algorithm
is called a \emph{kernel}.
Clearly, if $L$ is decidable and admits a kernel, 
$L$ is in FPT, and the converse also 
holds~(see e.g.~\cite[Theorem~4]{fomin2019kernelization}).

The class $W[1]$ contains all problems that 
admit FPT reductions from \textsc{Independent Set}
parameterized by the solution size, 
i.e. the number of vertices in the independent set.
Under the standard assumption that $\text{FPT} \neq \text{W}[1]$,
we can show that a problem is not fixed-parameter tractable
by proving its W[1]-hardness, i.e.
by providing an FPT reduction from \textsc{Independent Set} 
to the problem. The class \XP{} contains all parameterized problems
that can be solved in time $n^{f(k)}$ for instances with input size
$n$, parameter $k$ and some computable function $f$.

For sharper lower bounds, stronger assumptions are sometimes necessary. Here, we will primarily consider the {\em exponential-time hypothesis} which states that the 3-SAT problem is not solvable in {\em subexponential time} when parameterized by the number of variables $n$ or the number of clauses $m$. To make this more precise, for $d \geq 3$ let $c_d$ denote the infimum of all constants $c$ such that $d$-SAT is solvable in $2^{c n}$ time by a deterministic algorithm; then the ETH states that $c_3 > 0$. The {\em strong exponential-time hypothesis} (SETH) additionally conjectures that the limit of the sequence $c_3, c_4, \ldots$ tends to $1$, which in particular is known to imply that the satisfiability problem for clauses of arbitrary length (CNF-SAT)  is not solvable in $2^{c n}$  time for {\em any} $c < 1$~\cite{impagliazzo2009}.

\paragraph{Quantified Boolean Formulas.}

Boolean expressions, formulas, variables, 
literals, conjunctive normal form (CNF) and clauses are
defined in the standard way~(cf.~\cite{biere2021handbook}).
We treat $1$ and $0$ as the truth-values
``true'' and ``false'', respectively, 
and clauses in CNF as sets of literals.
We will assume that
no clause contains a literal twice or a literal and its negation, and no clause is repeated in any formula. 
A \emph{quantified Boolean formula} 
is of the form $\cQ. \phi$, where 
$\cQ = Q_1 x_1 \dots Q_n x_n$ with
$Q_i \in \{\forall, \exists\}$ for all $1 \leq i \leq n$
is the \emph{(quantifier) prefix},
$x_1, \dots, x_n$ are variables, and
$\phi$ is a propositional Boolean formula $\phi$
on these variables
called the \emph{matrix}.
If $Q_i = \forall$, we say that
$x_i$ is a \emph{universal variable},
and if $Q_i = \exists$, we say that
$x_i$ is an \emph{existential variable}.

For a Boolean formula $\phi$, a variable $x$
and value $b \in \{0,1\}$,
define $\phi[x = b]$
to be the formulas obtained from $\phi$
by replacing every occurrence of $x$ with $b$.
The truth value of a formula $\cQ. \phi$ is 
defined recursively.
Let $\cQ = Q_1 x_1 \dots Q_n x_n$ and
$\cQ' = Q_2 x_2 \dots Q_n x_n$.
Then $\cQ. \phi$ is \emph{true} if
\begin{itemize}
  \item $n = 0$ and $\phi$ is a true Boolean expression, or
  \item $Q_1 = \exists$ and 
  $\cQ'. \phi[x_1 = 0]$ \emph{or} $\cQ'. \phi[x_1 = 1]$ is true, or
  \item $Q_1 = \forall$ and 
  $\cQ'. \phi[x_1 = 0]$ \emph{and} $\cQ'. \phi[x_1 = 1]$ are true.
\end{itemize}
Otherwise, $\cQ. \phi$ is \emph{false}.

A formula $\cQ. \phi$ is a QCNF if $\phi$ is in CNF,
and a Q$d$-CNF if $\phi$ is in $d$-CNF, i.e. a CNF where every clause
has size at most $d$.
For a clause $C$ in a QCNF, we use
$C^\exists$ and $C^\forall$ to denote
the restriction of $C$ to existential
and universal variables, respectively.
If $\phi$ is in CNF, then $\phi[x=0]$
can be simplified by removing every clause
containing literal $\overline{x}$, and 
removing literal $x$ from the remaining clauses.
Analogously, $\phi[x=1]$ is simplified by
removing every clause containing 
the literal $x$, and removing the literal 
$\overline{x}$ from the remaining clauses.

Let $\cQ. \phi$ be a formula such that
$Q_1 = \dots = Q_i = \forall$ and
$Q_{i+1} = \dots = Q_n = \exists$
for some $1 \leq i \leq n$.
Then we write that $\cQ. \phi$ is
a $\forall \exists$BF,
and replace BF by CNF and $d$-CNF
if $\phi$ is in CNF and $d$-CNF,
respectively.

Following the convention in the literature,
we write \QBF{} for the problem of deciding whether a QCNF is true, and by \dQBF{}, \FEQBF{} and \FEdQBF{d} we denote the same problem
restricted to Q$d$-CNF, $\forall \exists$CNF and $\forall \exists d$-CNF formulas, respectively.

\section{Algorithms for \QBF{}} \label{sec:algorithms}

We show that \QBF{} parameterized by the
number of existential variables and the maximum size of any clause is linear-time fixed-parameter tractable, i.e., with a running time linear in the size of formula.
The algorithm is based on a chain of reductions involving two novel natural problems and for which we provide corresponding fixed-parameter algorithms along the way.
We begin by defining the \TAUT{} problem.

\pbDef{\TAUT{}}
{A set of variables $X$ and $d$-CNF formulas 
$\phi_1, \dots, \phi_k$ on $X$.}
{Is $\phi_{1} \lor \dots \lor \phi_{k}$ a tautology?}

The following lemma follows via an application of quantifier elimination.
\ifshort
\begin{lemma}[$\star$]
\fi
\iflong
\begin{lemma}
\fi
\label{lem:qbftocgis}
  There is an algorithm that takes a Q$d$-CNF
  $\cQ. \phi$ with $k$ existentially quantified variables,
  and in time $\bigoh(2^k |\phi|)$
  constructs an instance $I$ of \TAUT{}
  with $2^k$ $d$-CNF formulas of size at most $|\phi|$
  such that $\cQ. \phi$ holds if and only if
  $I$ is a yes-instance.
\end{lemma}
\iflong
\begin{proof}
  We will define a procedure
  that takes as input 
  a quantifier prefix $\cQ$ on variables $X$ 
  and a set of $d$-CNF formulas $D$ on $X$,
  and in time $\bigoh(|D| \cdot |\phi|)$
  computes a new quantifier prefix $\cQ'$
  on a new set of variables variables $X'$
  and a new set of $d$-CNF formulas $D'$ on $X'$
  such that
  \begin{enumerate}[(1)]
    \item \label{cond:numvars} 
    $|X'| \leq 2|X|$,
    \item \label{cond:numforms} 
    $|D'| \leq 2|D|$,
    \item \label{cond:formsize} 
    $|\phi'| \leq |\phi|$ for all $\phi' \in D'$ and $\phi \in D$,
    \item \label{cond:existless} 
    there is one less existential quantifier in $Q'$
    than in $Q$, and
    \item \label{cond:correct} 
    the formula $\cQ. \bigvee_{\phi \in D} \phi$
    is equivalent to
    $\cQ'. \bigvee_{\phi' \in D'} \phi'$.
  \end{enumerate}
  
  If the initial Q$d$-CNF is $\cQ. \phi$,
  we start with $\cQ$ and $D = \{\phi\}$,
  and recursively apply the procedure $k$ times
  until we obtain an equivalent formula
  $\cQ^\star. \bigvee_{\phi^\star \in D^\star} \phi^\star$, 
  where $\cQ^\star$ contains only universal quantifiers,
  the number of variables
  $|X^\star| \leq 2^k|X|$, and the number of formulas
  $|D^\star| \leq 2^k|D| \leq 2^k$.
  The total running time sums up to $\bigoh(2^k |\phi|)$.
  Since all quantifiers in $\cQ^\star$ are universal,
  $\cQ. \phi$ holds if and only if
  $\bigvee_{\phi^\star \in D^\star} \phi^\star$ is a tautology,
  i.e. $(X^\star, D^\star)$ is a yes-instance of \TAUT{}.

  Now we define the procedure.
  Let $\cQ = Q_1 x_1 \dots Q_n x_n$ 
  be the quantifier prefix,
  and $D$ be a set of $d$-CNF formulas.
  Let $i$ be the index of the last existential quantifier in $\cQ$,
  and observe that $Q_j = \forall$ for all $j > i$.
  Create a new quantifier prefix consisting of three parts
  $\cQ' = \cQ_{<i} \cQ_{0} \cQ_{1}$,
  where 
  $\cQ_{<i} = Q_1 x_1 \dots Q_{i-1} x_{i-1}$
  is a copy of $\cQ$ up to index $i$, and
  $\cQ_{b} = \forall x^{b}_{i+1} \dots \forall x^{b}_{n}$
  for both $b = 0$ and $b = 1$.
  For a formula $\phi$ with variables $x_1,\dots,x_n$ and 
  $b \in \{0,1\}$,
  let $\phi^{i,b}$ denote
  the formula obtained from $\phi$
  by replacing every variable $x_j$ with $j > i$ by $x^b_j$.
  For every $\phi \in D$,
  add $\phi[x_i=0]^{i,0}$ and $\phi[x_i=1]^{i,1}$ to $D'$.

  Clearly, the procedure can be implemented in $\bigoh(|D'|)$ time,
  and the resulting $\cQ'$ and $D'$ satisfy conditions
  \eqref{cond:numvars}, \eqref{cond:numforms}, 
  \eqref{cond:formsize} and \eqref{cond:existless}.
  It remains to show that~\eqref{cond:correct} also holds, 
  i.e. that the new formula is equivalent to the original one.
  To this end, let $\phi_D = \bigvee_{\phi \in D} \phi$ 
  and observe that
  \begin{align*}
    &\cQ.     &&\bigvee_{\phi \in D} \phi                                    &&\iff \\
    &\cQ_{<i} &&\exists x_i \forall x_{i+1} \dots x_{n}. \phi_D              &&\iff \\
    &\cQ_{<i}.&&(\forall x_{i+1} \dots x_{n}. \phi_D[x_i=0]) \lor \\
    &         &&(\forall x_{i+1} \dots x_{n}. \phi_D[x_i=1])                 &&\iff \\
    &\cQ_{<i}.&&\forall x^{0}_{i+1} \dots x^{0}_{n}. \phi_D[x_i=0]^{i,0} \lor \\
    &         && \forall x^{1}_{i+1} \dots x^{1}_{n}. \phi_D[x_i=1]^{i,1}    &&\iff \\
    &\cQ_{<i} &&\forall x^{0}_{i+1} \dots x^{0}_{n} \forall x^{1}_{i+1} \dots x^{1}_{n}. \\
    &         &&\phi_D^{i,0}[x_i=0] \lor \phi_D^{i,1}[x_i=1]                 &&\iff \\
    &\cQ'.    &&\bigvee_{\phi' \in D'} \phi',  \\
  \end{align*}  
  where the last equivalence follows from the expansion
  $\phi_D^{i,b}[x_i = b] = \bigvee_{\phi \in D} \phi^{i,b}[x_i = b]$.
\end{proof}
\fi

As the next step, we provide a polynomial-time reduction from
\TAUT{} to the problem of finding an independent set of size $k$
in a well-structured graph.
For a set of variables $X$ and integer $d$,
let $\cC_X^d$ be the set of all
clauses with exactly $d$ distinct literals 
over variables in $X$.
A pair $(G,\lambda)$ is a
\emph{$k$-partite $d$-clause graph over $X$}
if $G$ is an undirected $k$-partite graph with $V(G) = V_1 \uplus \cdots \uplus V_k$,
and $\lambda : V(G) \to \cC_X^d$ is a function, injective on $V_i$ for every $i \in [k]$.
Moreover, two vertices $u \in V_i$ and $v \in V_j$ in $G$
are connected by an edge if and only if
$i \neq j$ and $\lambda(u)$ and $\lambda(v)$ {\em clash},
i.e. they contain a pair of opposite literals.
The problem of finding an independent set in such 
a graph can be formally defined as follows. 
\pbDef{\CGIS{}}
{A $k$-partite $d$-clause graph $(G,\lambda)$ over $X$.}
{Is there an independent set $S \subseteq V(G)$ such that
$|S \cap V_i| = 1$ for all $1 \leq i \leq k$?}
In the following we will use $\lambda^{-1}$ as the inverse of
$\lambda$, which is well-defined if the part $V_i$ is clear from the
context since $\lambda$ is injective on every part $V_i$.
We also use $\lambda(S)$ to denote the set
$\{ \lambda(v) : v \in S \}$ for every set $S \subseteq V$.

\begin{lemma} \label{lem:tauttocgis}
  There is a linear-time reduction
  that takes an instance
  $I = (X, \{\phi_1, \dots, \phi_k\})$ of \TAUT{}, where each $\phi_i$ is a $d$-CNF, 
  and produces an instance 
  $I' = (G, X, \lambda)$ of \CGIS{} 
  where $(G,\lambda)$ is a $k$-partite $d$-clause graph over $X$ with the $i$-th part having at most $|\phi_i|$ vertices such that
  $I$ is a no-instance if and only if $I'$ is a yes-instance.
\end{lemma}
\begin{proof}
  Given an instance $I$ of \TAUT{}, 
  construct the $k$-partite $d$-clause 
  graph $(G,\lambda)$ over $X$ by letting 
  $V(G) = V_1 \uplus \cdots \uplus V_k$,
  where $V_i$ contains one vertex $v$
  for every clause $C$ in $\phi_i$,
  and setting $\lambda(v) = C$; by the assumption
  that all clauses in $\phi_i$ are distinct,
  this definition ensures that
  $\lambda$ is injective on every part $V_i$.
  For every $1 \leq i < j \leq n$
  and every $u \in V_i$ and $v \in V_j$,
  add an edge $\{u, v\}$ to $G$ if
  clauses $\lambda(u)$ and $\lambda(v)$ clash.
  To ensure that every
  clause $\lambda(v)$ contains exactly $d$ literals, we add $d-1$ new
  variables to $X$ and add $d-|\lambda(v)|$ of those positively to the
  clause $\lambda(v)$. Note that this does not modify the edge set of the graph because all new variables only occur positively.  The ideas behind the reduction is illustrated
  in Figure~\ref{fig:tautgraph}.
  Clearly, this reduction requires polynomial time.

  \begin{figure}
    \begin{center}
      \begin{tikzpicture}[yscale=0.65, inner sep=1pt]
        \draw (0,0) ellipse (0.7cm and 3cm);
        \node[shape=circle,draw=black] (U1) at (0,2) {$u_{1}$};
        \node[shape=circle,draw=black] (U2) at (0,0.75) {$u_{2}$};
        \node[shape=circle,draw=black,fill=gray!40] (U3) at (0,-0.75) {$u_{3}$};
        \node[shape=circle,draw=black] (U4) at (0,-2) {$u_{4}$};
        \draw (2,0) ellipse (0.7cm and 3cm);
        \node[shape=circle,draw=black] (V1) at (2,2) {$v_{1}$};
        \node[shape=circle,draw=black] (V2) at (2,0) {$v_{2}$};
        \node[shape=circle,draw=black,fill=gray!40] (V3) at (2,-2) {$v_{3}$};
        \draw (4,0) ellipse (0.7cm and 3cm);
        \node[shape=circle,draw=black,fill=gray!40] (W1) at (4,2) {$w_{1}$};
        \node[shape=circle,draw=black] (W2) at (4,0) {$w_{2}$};
        \node[shape=circle,draw=black] (W3) at (4,-2) {$w_{3}$};
        \draw (6,0) ellipse (0.7cm and 3cm);
        \node[shape=circle,draw=black] (Z1) at (6,2) {$z_{1}$};
        \node[shape=circle,draw=black] (Z2) at (6,0.75) {$z_{2}$};
        \node[shape=circle,draw=black,fill=gray!40] (Z3) at (6,-0.75) {$z_{3}$};
        \node[shape=circle,draw=black] (Z) at (6,-2) {$z_{4}$};
        \draw (U1)--(V1);
        \draw (U1)--(V3);
        \draw (U1) to[out=20,in=160] (W1);
        \draw (U1)--(W2);
        \draw (U1) to[out=25,in=160] (Z1);
        \draw (U1)--(Z2);
      \end{tikzpicture}
    \end{center}
    \caption{Let $X=\{x_1,x_2,\ldots,x_6\}$ be a set of variables. Let $\phi_1=(x_1\vee x_2 \vee x_3)\wedge (\overline{x}_1\vee x_2 \vee x_4)\wedge (\overline{x}_1\vee x_5 \vee \overline{x}_6) \wedge (\overline{x}_3\vee x_2 \vee x_5)$, 
      $\phi_2=(x_2 \vee \overline{x}_3 \vee \overline{x}_6)\wedge (\overline{x}_4\vee \overline{x}_5 \vee x_6)\wedge (\overline{x}_2\vee x_3 \vee \overline{x}_6)$, 
      $\phi_3=(\overline{x}_2 \vee {x}_3 \vee {x}_4)\wedge (\overline{x}_2\vee \overline{x}_4 \vee x_5)\wedge ({x}_3\vee x_4 \vee \overline{x}_5)$, and 
      $\phi_4=(x_1\vee \overline{x}_2 \vee x_5)\wedge
      (\overline{x}_3\vee x_4 \vee x_6)\wedge (x_3\vee x_4 \vee
      \overline{x}_6) \wedge (\overline{x}_4\vee \overline{x}_5 \vee
      x_6)$. The vertices $u_1,\ldots,u_4$ corresponds to the first,
      second, third, and fourth clauses of $\phi_1$,
      respectively. \iflong Similarly, the vertices $v_1,v_2,v_3$
        corresponds to the first second, and third clauses of
        $\phi_2$, respectively. \fi Analogously, we have drawn
      vertices for clauses in \ifshort $\phi_2$, \fi $\phi_3$ and $\phi_4$. All the edges incident on $u_1$ are drawn in the figure. The set $\{u_3,v_3,w_1,z_3\}$ is an independent set the graph and the assignment $\alpha(x_1)=\alpha(x_2)=\alpha(x_6)=1$ 
      and $\alpha(x_3)=\alpha(x_4)=\alpha(x_5)=0$
      implies that $\bigvee_{i=1}^4\phi_i$ is not a tautology.
    }\label{fig:tautgraph}
  \end{figure}

  For correctness, first assume that $I$ is a no-instance,
  i.e. $\phi_1 \lor \dots \lor \phi_k$ is not a tautology.
  Then there exists an assignment $\alpha : X \to \{0,1\}$ 
  that falsifies every formula $\phi_i$.
  For each $1 \leq i \leq k$, let $C_i$ be a clause in $\phi_i$
  falsified by $\alpha$,
  and let $v_i \in V_i$ be the vertex of $G$ 
  such that $\lambda(v_i) = C_i$.
  Observe that $\alpha$ satisfies 
  $\bigwedge_{i=1}^{k} \neg C_i = \bigwedge_{\ell \in C_1 \cup \dots C_k} \overline{\ell}$,
  hence no pair of clauses $C_i$ and $C_j$ clash.
  By construction of $G$, this implies that
  there is no edge between $v_i$ and $v_j$  
  for any $i$ and $j$, i.e.
  $\{v_i : 1 \leq i \leq k\}$ is an independent set in $G$.

  Now suppose $I'$ is a yes-instance,
  i.e. $G$ contains an independent set 
  $S = \{v_1,\dots,v_k\}$,
  where $v_i \in V_i$ for all $1 \leq i \leq k$.
  By construction of $(G, \lambda)$,
  no two clauses $\lambda(v_i)$ and $\lambda(v_j)$ clash.
  Observe that $\phi' = \bigwedge_{i=1}^{k} \neg \lambda(v_i)$
  is a conjunction of literals, and it contains
  no two opposite literals,
  hence $\phi'$ is satisfiable.
  Any assignment $\alpha' : X \to \{0,1\}$ that satisfies
  $\phi'$ falsifies at least one clause
  in every formula $\phi_i$, namely $\lambda(v_i)$. 
  Thus, $\alpha'$ falsifies $\phi_1 \lor \dots \lor \phi_k$,
  proving that it is not a tautology.
\end{proof}

The following theorem shows that \CGIS{} is in \XP{} parameterized by
$k$ only.
\ifshort
\begin{theorem}[$\star$]
\fi
\iflong
\begin{theorem}
\fi
\label{the:cgis-xp}
  \CGIS{} can be solved in time
  $\bigoh((\max_{i=1}^k|V_i|)^kd\binom{k}{2})$ and is therefore
  in \XP{} parameterized by $k$.
\end{theorem}
\iflong
\begin{proof}
  Let $G$ be a $k$-partite $d$-clause graph $G$ over $X$ with
  $k$-partition $V_1 \uplus \cdots \uplus V_k$,
  A simple brute-force algorithm that enumerates all possible
  $k$-tuples $(v_1,\dotsc,v_k)$ such that $v_i \in V_i$ and for each
  tuple checks in polynomial-time whether $\{v_1,\dotsc,v_k\}$ is an
  independent set in $G$ runs in $\bigoh((\max_i {|V_i|})^k d \binom{k}{2})$ time.
\end{proof}
\fi

Now we will show that \CGIS{} is in FPT parameterized by $k$ and $d$.
To this end, we will use the sunflower lemma.
For a family $\cF$ of sets over some universe $U$, we say that a subset
$\cS \subseteq \cF$ is a \emph{sunflower} if there is a subset $C \subseteq U$ such that $F \cap F'= C$ for every two distinct
$F,F' \in \cS$, i.e. all pairs of distinct sets in $\cS$ have a common
intersection $C$, which we also call the \emph{core} of the sunflower. 
If $\cS$ is a sunflower with core $C$ and $F \in \cS$, then we call 
$F \setminus C$ the \emph{petal} of $F$.
Observe that the petals $\{ F \setminus C \mid F \in \cS \}$ 
of a sunflower are pairwise disjoint.

\begin{lemma}[\cite{Erdos60}]\label{lem:SF}
  Let $\cF$ be a family of subsets of a universe $U$, each of size exactly
  $b$, and let $a \in \mathbb{N}$. If $|\cF|\geq b!(a-1)^{b}$, then $\cF$
  contains a sunflower $\cS$ of size at least $a$. Moreover,
  $\cS$ can be computed in time $\bigoh(|\cF|d)$.
\end{lemma}

For a graph $G$ and a vertex $v$, we write
$G - v$ to denote the graph obtained from $G$
by deleting the vertex $v$ with all incident edges.

\begin{lemma}\label{lem:applySF}
  Let $(G,\lambda)$ be a $k$-partite $d$-clause graph, and
  $V_i$ be one of the parts.
  If the family $\lambda(V_i)$ 
  contains a sunflower $\cS$ 
  of size at least $s=(k-1)d+2$, then,
  for every $v$ with $\lambda(v) \in \cS$,
  instances $(G,\lambda,X)$ and $(G - v,\lambda,X)$ of \CGIS{}
  are equivalent.
\end{lemma}
\begin{proof}
  Let $\cS$ be the sunflower in $\lambda(V_i)$ of size at least $s$,
  let $F \in \cS$ be an arbitrary clause in $\cS$, and 
  $v = \lambda^{-1}(F)$ be the corresponding vertex in $G$. 
  We claim that $v$ satisfies the statement of the
  lemma, i.e. $(G,\lambda)$ has an independent set $S$ with $|S\cap
  V_i|=1$ if and only if so does $(G - v,\lambda)$.
  The reverse direction of the claim is clear
  since $G - v$ is a subgraph of $G$. 
  
  Towards showing the forward direction, 
  let $S$ be a solution to $(G, \lambda, X)$. 
  If $v \notin S$, then $S$ is also 
  a solution to $(G-v, \lambda, X)$.
  Now suppose that $v \in S$.
  We claim that there exists 
  a clause $F' \in \cS$
  such that $F'$ does not clash with any clause 
  $\lambda(u)$ for $u \in S \setminus \{v\}$,
  so we can replace $v$ with $v' = \lambda^{-1}(F')$
  in the independent set $S$.
  To this end, let $S' = S \setminus v$ and
  let $C$ be the core of $\cS$.
  Note that every clause
  contains exactly $d$ literals,
  therefore it can share variables 
  with the petals of at most $d$ clauses in $\cS$.
  Thus, the clauses in $\lambda(S')$
  share variables with at most $|S'| d = (k-1)d$ 
  petals of $\cS$ in total.
  Since $|\cS \setminus \{F\}| \geq s-1 = (k-1)d + 1$,
  there is a clause $F' \in \cS \setminus \{F\}$ 
  whose petal $F' \setminus C$ shares no variables with
  any clause in $\lambda(S')$.
  Moreover, the core $C \subseteq F$ does not clash with any
  clause in $\lambda(S')$ since $v = \lambda^{-1}(F)$
  is not adjacent to any vertex in $S'$.
  Therefore, $v' = \lambda^{-1}(F')$ is not adjacent to 
  any vertex in $S'$ and $S' \cup \{v'\}$ is
  an independent set in $G-v$.
\end{proof}

\begin{theorem}\label{the:cgis-kernel}
  \CGIS{} has a kernel with at most $d!(s-1)^d - 1$ vertices in every part
  $V_i$, where $s=(k-1)d+2$. Note that this implies that \CGIS{} has a kernel of
  size at most $dkd!(s-1)^d$. The kernel can be computed in time 
  $\bigoh(\min\{d!(s-1)^{d}d|V(G)|, d|V(G)|^2\})$.
\end{theorem}
\begin{proof}
  Let $(G,\lambda)$ be a $k$-partite $d$-clause graph with
  parts $V_1,\dotsc,V_k$ over some variables $X$. 
  If $|V_i|< d!(s-1)^d$ for every $i \in [k]$, 
  where $s=(k-1)d+2$, then the instance is already
  kernelized. 
  So suppose that $|V_i|\geq d!(s-1)^d$.
  By Lemma~\ref{lem:SF}, $\lambda(V_i)$ has a sunflower $\cS$ of size at least
  $s$, and $\cS$ can be found in polynomial time. Note that by taking
  any subset $V' \subseteq \lambda(V_i)$ of size at least $d!(s-1)^d$,
  $\cS$ can be found in time 
  $\bigoh(\min\{d!(s-1)^{d}d, d|V(G)|\})$.
  Then, by Lemma~\ref{lem:applySF}, we can remove any vertex in
  $\lambda^{-1}(\cS)$ from $V_i$ and obtain an equivalent but smaller
  instance. 
  Therefore, applying this procedure exhaustively, we obtain in
  polynomial time an equivalent instance of \CGIS{} such that
  $|V_i|<d!(s-1)^d$ for every $1 \leq i \leq k$. 
  Notice that we need to apply the above procedure $\bigoh(|V(G)|)$ times and hence the running time is at most 
  $\bigoh(\min\{d!(s-1)^{d}d|V(G)|, d|V(G)|^2\})$.
\end{proof}

\begin{corollary}\label{cor:cgis-fpt}
  \CGIS{} is fixed-parameter tractable parameterized by $k+d$.
  In particular, it can be solved in time
  $(dk)^{\bigoh(dk)}|V(G)|$.
\end{corollary}
\begin{proof}
  The statement that \CGIS{} is fixed-parameter tractable
  parameterized by $k+d$ follows immediately from
  \Cref{the:cgis-kernel}. Let $(G,\lambda)$ be a $k$-partite
  $d$-clause graph with parts $V_1,\dotsc,V_k$ over some variables
  $X$. First we can employ \Cref{the:cgis-kernel} to obtain the kernel
  $(G',\lambda')$ with parts $V_1',\dotsc,V_k'$ in time
  $d!(s-1)^{d}d|V(G)|$, where
  $s=(k-1)d+2$, that is equivalent to
  $(G,\lambda)$ and satisfies $|V_i'|\leq d!(s-1)^d$. We can then use
  \Cref{the:cgis-xp} to solve the instance $(G',\lambda')$ of \CGIS{} in time
  $\bigoh( (\max_{i=1}^k|V_i'|)^kd\binom{k}{2})
  =\bigoh( (d!)^{k} (s-1)^{dk}d\binom{k}{2})$,
  which is bounded from above by
  $(dk)^{\bigoh(dk)}$.
  Altogether, we therefore obtain
  $(dk)^{\bigoh(dk)} + \bigoh(d!(s-1)^{d}d|V(G)|)$, which is bounded by $(dk)^{\bigoh(dk)}|V(G)|$ as the
  running time of our algorithm.
\end{proof}

Combining \Cref{lem:tauttocgis} and \Cref{lem:qbftocgis}, we obtain:
\begin{corollary}\label{cor:qbftocg}
  There is a reduction that takes a Q$d$-CNF
  $\cQ. \phi$ with $k$ existentially quantified variables,
  and in time $\bigoh(2^k |\phi|)$
  produces an instance 
  $I' = (G, \lambda, X )$ of \CGIS{} 
  where $(G,\lambda)$ is a $2^k$-partite $d$-clause graph over $X$ with each part having at most $|\phi|$ vertices
  such that
  $I$ is a no-instance if and only if $I'$ is a yes-instance.
\end{corollary}

We now show that \QBF{} with $k$ existential
variables and clauses of size $d$
is in FPT parameterized by $k + d$.

\begin{theorem}
  \label{thm:QBF-FPT}
  \QBF{} is fixed-parameter tractable parameterized by the number $k$
  of existential variables plus the maximum size $d$ of any clause.
  In particular, there is an algorithm solving this problem in
  time $(2^kd)^{\bigoh(2^kd)}|\phi|$
  for a QCNF formula $\cQ.\phi$.
\end{theorem}
\begin{proof}

  Let $\Phi=\cQ.\phi$ be the given \QBF{} formula with $k$ existential
  variables having clauses of size at most $d$.
  We first use \Cref{cor:qbftocg} to obtain in time $\bigoh(2^k
  |\phi|)$ the
  $2^k$-partite $d$-clause graph $(G,\lambda)$ with parts
  $V_1,\dotsc,V_{2^k}$ such that $\Phi$ is false if and only if $(G,\lambda)$ has an independent set $S$
  with $|S\cap V_i|=1$. We then use \Cref{cor:cgis-fpt} to decide in
  time $(d2^k)^{\bigoh(d2^k)}|V(G)|=(d2^k)^{\bigoh(d2^k)}2^k|\phi|=(d2^k)^{\bigoh(d2^k)}|\phi|$
  whether $(G,\lambda)$ has an independent set $S$
  with $|S\cap V_i|=1$. If so, we return that $\Phi$ is false, 
  otherwise we return that $\Phi$ is true.
  Altogether the total runtime of the algorithm is at most
  $(d2^k)^{\bigoh(d2^k)}|\phi|$.
\end{proof}

Last, via a straightforward algorithm we remark that \QBF{} is in \XP{} parameterized by the number of existential variables. As we will see in Section~\ref{sec:lower_bounds} this problem is unlikely to admit an FPT algorithm unless FPT=W[1].

\ifshort
\begin{theorem}[$\star$]
\fi
\iflong
\begin{theorem}
\fi
  \QBF{} is in \XP{} parameterized by the number $k$ of existential
  variables. In particular, \QBF{} can be solved in time
  $\bigoh(m^{2^k}d\binom{k}{2}+2^k|\phi|)$ for a QCNF-formula $\cQ.\phi$ with $m$ clauses.
\end{theorem}
\iflong
\begin{proof}
  Let $\Phi=\cQ.\phi$ be the given \QBF{} formula with $k$ existential
  variables having clauses of size at most $d$.
  We first use \Cref{cor:qbftocg} to obtain in time $\bigoh(2^k
  |\phi|)$ the
  $2^k$-partite $d$-clause graph $(G,\lambda)$ with parts
  $V_1,\dotsc,V_{2^k}$ such that $\Phi$ is not
  satisfiable if and only if $(G,\lambda)$ has an independent set $S$
  with $|S\cap V_i|=1$. We then use \Cref{the:cgis-xp} to solve the
  instance $(G,\lambda)$ of \CGIS{} in time
  $\bigoh(m^{2^k}d\binom{k}{2})$. Therefore, we obtain
  $\bigoh(m^{2^k}d\binom{k}{2}+2^k|\phi|)$ as the total runtime of the
  algorithm.
\end{proof}
\fi

\section{The QCSP Problem}
\label{sec:qcsp}

We now consider the finite-domain generalization of \QBF{} known as the \emph{quantified constraint satisfaction problem (QCSP)}. 
An instance of the \emph{constraint satisfaction problem (CSP)} (without quantifiers) is $(X, D, C)$, 
where $X = \{x_1,\dots,x_n\}$ is a set of variables, 
$D = \{D_1, \dots, D_n\}$ is a set of \emph{domains (of values)} for each variable, and 
$C = \{C_1, \dots, C_m\}$ is a set of constraints, where $C_j = R_j(x_{j_1}, \dots, x_{j_{\ar(R_j)}})$, $R_j \subseteq D_{j_1} \times \dots \times D_{j_{\ar(R_j)}}$ is a relation of arity $\ar(R_j)$, and
$1 \leq j_1, \dots, j_{\ar(R_j)} \leq n$. 
The instance is \emph{satisfiable} if there exists an assignment $\alpha : X \to \bigcup_{i=1}^{n} D_i$ of values to the variables such that $\alpha(x_i) \in D_i$ for all $i \in [n]$, and 
$(\alpha(x_{j_1}), \dots, \alpha(x_{j_{\ar(R_j)}})) \in R_j$ for all constraints in $j \in [m]$.
Let $d = \max_{i=1}^{n} |D_i|$ be the largest domain size and $r = \max_{j=1}^{m} \ar(R_j)$ be the maximum arity of a constraint in $C$.
For example, 3-SAT can be cast as CSP with $d=2$ (every variable is assigned a Boolean value) and $r=3$ (every clause is a ternary constraint).

Parameterized complexity of the CSP with respect to $n$, $d$ and $r$ is well-understood~\cite{samer2010constraint}:
the problem is in FPT parameterized by $n + d$, W[1]-hard parameterized by $n + r$, in XP parameterized by $n$, and paraNP-hard parameterized by $d+r$.

QCSP is a generalization where the input additionally comes with a set of quantifiers $\cQ = (Q_1, \dots, Q_n)$.
The basic notions in Section~\ref{sec:preliminaries} easily extends to the CSP setting and we write QCSP for the (PSPACE-complete) decision problem of verifying whether an instance $\cQ . \phi$ is true or false, where $\phi$ is a CSP instance over the variables occurring in the quantifier prefix $\cQ$. Naturally, if a variable $x_i$ has domain $D_i$ then we in the context of a universal quantifier require that the subsequent formula is true for all values in $D_i$,
and for at least one value in $D_i$ if the quantifier is existential.
We manage to generalize Theorem~\ref{thm:QBF-FPT} to QCSP.

\ifshort
\begin{theorem}[$\star$]
\fi
\iflong
\begin{theorem}
\fi
 QCSP is FPT when parameterized by the number of existentially quantified variables, the domain, and the maximum arity of any relation.
\end{theorem}

\iflong 
\begin{proof}
   Let 
\[\Phi=Q_1  x_1,  \ldots Q_n x_n . \phi\]

where each $Q_i \in \{\forall, \exists\}$ be an instance of QCSP over variables $V = \{x_1, \ldots, x_n\}$, domain values $D = \{D_1, \ldots, D_n\}$, and constraints $C$. We write $\Gamma = \{R \mid R(\mathbf{x}) \in C\}$ for the set of relations in the instance and let
\[r = \max \{\mathrm{ar}(R) \mid R \in \Gamma\}\] be the maximum arity of any relation. Last, let
$d =  \lceil \log_2 |D| \rceil$. 
Define the surjective function $h \colon \{0,1\}^d \to D$ such that there exists a unique element $m \in D$ where $|h^-1(m)| \geq 1$ and where we for every other $a \in D$ have $|h^-1(a)| = 1$ Hence, every domain value in $D$ is represented by a Boolean $d$-ary tuple, and if $2^d > |D|$ then a unique value in $D$ corresponds to the additional tuples in $\{0,1\}^d$. %

We will show an fpt reduction to QBFSAT parameterized by arity $q = r \cdot d$ and the number of existentially quantified variables. First, for an $n$-ary relation $R \in \Gamma$ we let $R_{\mathbb{B}} = \{(x^1_1, \ldots, x^1_d, \ldots, x^n_1, \ldots, x^n_d) \mid (h(x^1_1, \ldots, x^1_d), \ldots, h(x^n_1, \ldots, x^n_d))\}$ be the Boolean relation obtained by viewing each domain value in $D$ as a $d$-ary Boolean tuple via the surjective function $h$. Importantly, it is not hard to see that $R_{\mathbb{B}}$ can be defined by a conjunction of $(d \cdot \mathrm{ar}(R))$-clauses over Boolean variables $x_1, \ldots, x_{d \cdot \mathrm{ar}(R)}$: for each $(b_1, \ldots, b_{d \cdot \mathrm{ar}(R)}) \notin R_{\mathbb{B}}$ simply add the clause $(\neg x_1 \lor \ldots \lor \neg x_{d \cdot \mathrm{ar}(R)})$. Furthermore, we observe that any clause of arity smaller than $d \cdot r$ can be simulated by a $(d \cdot r)$-clause by repeating one of its arguments.

Now, let 
\[\Phi=Q_1  x_1,  \ldots Q_n x_n . \phi\]

where each $Q_i \in \{\forall, \exists\}$ be an instance of QCSP$(\Gamma)$. Crucially, at most $k$ variables are existentially quantified. For each $x_i$ we introduce $d$ fresh variables $x^1_i, \ldots, x^d_i$ and observe that this in total requires $n \cdot d$ fresh variables but that at most $k \cdot d$ of these correspond to existentially quantified variables. For each constraint $R(x_{i_1}, \ldots, x_{i_{\mathrm{ar}(R)}})$ occurring in $\phi$ we replace it by the conjunction of $(d \cdot \mathrm{ar}(R))$-clauses defining $R_{\mathbb{B}}(x^1_{i_1}, \ldots, x^1_{i_{\mathrm{ar}(R)}}, \ldots, x^d_{i_1}, \ldots, x^d_{i_{\mathrm{ar}(R)}})$. We let 

\begin{align*}
\varPhi &= Q_1 x^1_1, \ldots, x^d_1 \ldots Q_n x^1_n, \ldots, x^d_n . \\
& \phi_{\mathbb{B}}(x^1_1, \ldots, x^d_1, x^1_n, \ldots, x^d_n)
\end{align*}

be the instance of Q$(d \cdot r)$-CNF resulting from replacing each constraint by the corresponding conjunction of $(d \cdot r)$-clauses, where we with a slight abuse of notation write $Q_i x^1_i, \ldots, x^d_i$ with the meaning that all variables $x^1_i, \ldots, x^d_i$ have the same quantifier $Q_i$. For each variable domain $D_i \in D$ we first remark that $D_i$ can be treated as a unary relation and that $D_{i_{\mathbb{B}}}$ is thus a well-defined $d$-ary Boolean relation. Hence, for every $x_i$ we add the set of $d$-clauses defining 
$D_{i_{\mathbb{B}}}$ over the variables $x^i_1, \ldots, x^i_d$. Last, for every Boolean variable $x^j_i$ we simply use $\{0,1\}$ as the allowed domain values. We let $\varPhi'$ be the resulting Q$(d \cdot r)$-CNF instance.

For correctness, assume that $\Phi$ is true and has a winning strategy which for every existential variable $x_i$ is witnessed by a function $f_i \colon D^j  \to D$ where $j$ is the number of universally quantified variables preceding $x_i$. We construct a winning strategy for the Boolean instance $\varPhi'$ as follows. For each existential variable $x_i$ let $f_{i,\mathbb{B}}$ be the $(j \cdot d)$-ary Boolean function defined as $f_{i, \mathbb{B}}(h^{-1}(a_1), \ldots, h^{-1}(a_j)) = f_i(a_1, \ldots, a_j)$ for all $a_1, \ldots, a_j \in D$. We observe that this correctly defines a total Boolean function since $h$ is a bijection and it is easy to see that it must be a winning strategy for $\varPhi$. The other direction can be proven with a similar argument. 
\end{proof} 
\fi
It is worth remarking that by Samer \& Szeider~\cite{samer2010constraint} and the forthcoming Theorem~\ref{thm:AECNF-w1-hard} each of the above three conditions are necessary in the sense that we obtain a W[1]-hard problem if any condition is dropped.

\section{Lower Bounds} \label{sec:lower_bounds}

We proceed by complementing our positive FPT result by two strong lower bounds. We begin by ruling out an FPT algorithm for \FEQBF{} parameterized by the number of existential variables for unbounded clause size via a reduction from the W[1]-complete problem \textsc{Multipartite Independent set}. Using the following auxiliary result, we strengthen this result even further under the ETH.

\begin{theorem}[\cite{chen2006strong}, cf. Theorem~14.21~in~\cite{parameterizedalgorithmsbook}]
  \label{thm:ETH-IS}
  Assuming the ETH, there is no algorithm that
  decides if a graph on $n$ vertices
  has an independent set of size $k$
  in $f(k) \cdot n^{o(k)}$ time
  for any computable function $f$.
\end{theorem}

For our purposes, it is more convenient to work with
the following variant of the \textsc{Independent Set} problem.
In \textsc{Multipartite Independent Set},
an instance is a graph with the vertex set partitioned into $k$ parts,
and the question is whether the graph contains
an independent set with one vertex from each part.
Using a well-known 
reduction~(cf. Section~13.2~in~\cite{parameterizedalgorithmsbook}) that
takes an instance $(G, k)$ of \textsc{Independent Set} and constructs in
polynomial time an equivalent instance of \textsc{Multipartite Independent Set}
with $|V(G)|k$ vertices and the same parameter $k$,
we obtain the following corollary.

\ifshort
\begin{corollary}[$\star$]
\fi
\iflong
\begin{corollary}
\fi
  \label{cor:k-Partite-IS-ETH}
  Assuming the ETH, there is no algorithm that
  solves \textsc{Multipartite Independent Set} in
  $f(k) \cdot n^{o(k)}$ time
  for any computable function $f$.
\end{corollary}
\iflong
\begin{proof}
  We provide a short proof for completeness.
  Let $G$ be a graph with vertices
  $v_1, \dots, v_n$.
  Create a $k$-partite graph $G'$
  with vertices $V_1 \uplus \dots V_k$,
  where $V_i = \{(i,j) : j \in [n]\}$;
  add edges
  $\{ (i,j), (i',j') \}$ to $E(G')$
  for all $i \neq i'$ and 
  $\{ v_j, v_{j'} \}$ in $E(G)$.
  It is easy to see that
  $G$ contains an independent set
  of size $k$ if and only if
  $G'$ contains an independent set
  with one vertex from each part $V_i$,
  i.e. $(G',k)$ is a yes-instance of
  \textsc{Multipartite Independent Set}
  (see Section~13.2~in~\cite{parameterizedalgorithmsbook} for a full proof).

  Now, suppose there is an algorithm that solves \textsc{Multipartite Independent Set} in time
  $f(k) \cdot |V(G')|^{o(k)}$.
  Since $|V(G')| = nk$, we can use it 
  to decide whether $G$ has
  an independent set of size $k$
  in $f(k) \cdot (nk)^{o(k)} = (f(k) k^{o(k)}) \cdot n^{o(k)}$ time
  plus the polynomial time of
  the reduction,
  which contradicts the ETH
  by Theorem~\ref{thm:ETH-IS}.
\end{proof}

\fi

\begin{lemma}
  \label{lem:k-Partite-IS-to-AECNF}
  There is a polynomial-time reduction that takes
  an instance $(G, k)$ of \textsc{Multipartite Independent Set}
  and produces in polynomial time
  a $\forall \exists \text{CNF}$ formula
  with $\ceil{\log_2 (k)}$ existential variables
  such that $(G, k)$ is a yes-instance
  if and only if the formula is false.
\end{lemma}
\begin{proof}
  Let $G$ be a graph with vertex set $V_1 \uplus \dots \uplus V_k$.
  It will be convenient to assume that $k$ is a power of two.
  To this end, let $\kappa = \ceil{\log_2(k)}$, and add
  $2^\kappa - k$ new parts to $V(G)$, each consisting of one isolated vertex.
  Clearly, the new instance is equivalent to the original one,
  so we assume from now on that $k = 2^\kappa$.
  Enumerate vertices in each part of the graph.
  For convenience, we refer to vertex $j$ in part $V_i$ as $(i, j)$.

  We will construct a formula $\forall Y \exists X. \phi$
  on variables $Y = \{ y_v : v \in V(G) \}$ and $X = \{x_1, \dots, x_\kappa\}$
  that is false if and only if $(G, k)$ has an independent set
  with exactly one vertex from every part.
  To this end, enumerate all functions $\alpha_1, \dots, \alpha_{k}$
  from $X$ to $\{0,1\}$.
  For every vertex $v = (i,j) \in V(G')$,
  add a clause $C_{v}$ to $\phi$ with the following literals:
  \begin{itemize}
    \item $y_v$ and $\overline{y_u}$ for all $u \in V(G) \setminus V_i$
    such that $\{u,v\} \in E(G)$,
    \item $x_\ell$ if $\alpha_i(x_\ell) = 0$ and $\overline{x_\ell}$ if $\alpha_i(x_\ell) = 1$ for all $\ell \in [k]$.
  \end{itemize}
  This completes the construction.

  Towards correctness, first assume that
  $S$ is an independent set in $G$
  with one vertex from each $V_i$.
  Consider the set of clauses $\cC_S = \{C^{\forall}_v : v \in S\}$;
  recall that $C^\exists$ and $C^\forall$ for a clause $C$ denotes the
  restriction of $C$ to existential and universal variables, respectively.
  We claim that there is an assignment
  that falsifies every clause in $\cC_S$.
  It suffices to show that no pair of clauses
  in $\cC_S$ clashes, i.e. contain opposite literals.
  Consider two clauses $C^{\forall}_u, C^{\forall}_v \in \cC_S$.
  Since $S$ is an independent set, $u$ and $v$ are not adjacent,
  so $C^{\forall}_u$ does not contain $y_v$
  and $C^{\forall}_v$ does not contain $y_u$.
  Furthermore, both $C^{\forall}_u \setminus \{y_u\}$ and
  $C^{\forall}_v \setminus \{y_v\}$ only contain
  negative literals, so they do not clash either.
  Now, let $\tau : Y \to \{0,1\}$ be an assignment that
  falsifies all clauses in $\cC_S$.
  We claim that $\exists X. \phi[\tau]$ is false,
  i.e. $\phi[\tau]$ is not satisfiable.
  Indeed, for every assignment $\alpha_i : X \to \{0,1\}$,
  there is a vertex $v = (i, j) \in S$ in $G$, 
  and hence a clause $C^{\exists}_{v}$ remains in $\phi[\tau]$,
  which excludes $\alpha_i$ as the satisfying assignment.
  Thus, all assignments are excluded, and $\exists X. \phi[\tau]$ is false.

  For the other direction,
  suppose $\forall Y \exists X. \phi$ is false.
  Then there exists an assignment $\tau' : Y \to \{0,1\}$
  such that $\exists X. \phi[\tau']$ is false, i.e.
  $\phi[\tau']$ is not satisfiable.
  By construction, every clause of $\phi[\tau']$
  is a $\kappa$-clause with literals over all variables $x_1, \dots, x_\kappa$.
  Each such clause excludes exactly one satisfying assignment,
  so $\phi[\tau']$ contains exactly $2^\kappa = k$ clauses.
  Thus, for every assignment $\alpha_i : X \to \{0,1\}$,
  there exists a clause $C_{v}$ where $v \in V_i$ and
  $\tau'$ falsifies $C^{\forall}_{v}$.
  Pick one such clause $C_v$ for every $i$, and 
  let vertices $v$ form a set $S$.
  We claim that $S$ is an independent set in $G$.
  Suppose towards contradiction that $u,v \in S$ and $\{u,v\} \in E(G)$.
  By construction, $y_u \in C^{\forall}_u$ and $\overline{y_u} \in C^{\forall}_v$,
  so $\tau'$ cannot falsify both of them,
  which contradicts our choice of $u,v$.
\end{proof}

\begin{theorem} \label{thm:AECNF-w1-hard}
  \FEQBF{} is 
  W[1]-hard parameterized by the number of existential variables.
  Moreover, assuming the ETH, 
  this problem cannot be solved in $f(k) \cdot |\phi|^{o(2^k)}$ time for any computable function $f \colon \mathbb{N} \to \mathbb{N}$, where $\cQ. \phi$ is the $\forall \exists \text{CNF}$ with $k$
  existential variables.
\end{theorem}
\begin{proof}
  W[1]-hardness is immediate from Lemma~\ref{lem:k-Partite-IS-to-AECNF}
  and the fact that \textsc{Multipartite Independent Set} is W[1]-hard.
  Moreover, if there is an algorithm deciding whether 
  a $\forall \exists \text{CNF}$ with matrix $\phi$ and 
  $\kappa$ existential variables is true or false 
  in $f(\kappa) \cdot |\phi|^{o(2^\kappa)}$ time for any computable function $f$, 
  then it can be combined with the reduction of Lemma~\ref{cor:k-Partite-IS-ETH}
  to solve \textsc{Multipartite Independent Set}
  in $f(\ceil{\log_2 k}) \cdot 2^{o(k)}$ time, contradicting the ETH.
\end{proof}

For lower bounds for \FEdQBF{3}, we first observe that this problem cannot be solved in $2^{o(k)} \cdot |\phi|^{\bigoh(1)}$ time under the ETH (since we have a trivial reduction from 3-SAT). However, we can significantly sharpen this under the SETH and in fact rule out every $2^{\bigoh(k)}$ time algorithm, i.e. the problem is not solvable in $c^k \cdot |\phi|^{\bigoh(1)}$ time for {\em any} constant $c \geq 1$.
We rely on the following result.

\begin{theorem}[\cite{CalabroIP13}]
  \label{thm:SETH-for-AE3SAT}
  Assuming the SETH, 
  \FEdQBF{3} is not solvable
  is $\bigoh^*(c^n)$ time\footnote{The notation
  $\bigoh^*$ hides factors polynomial in $n$.}
  for any $c < 2$.
\end{theorem}

\begin{theorem}
  Assuming the SETH,
  \FEdQBF{3} parameterized by the number of existentially quantified variables $k$ is not solvable in
  $c^k \cdot |\phi|^{\bigoh(1)}$ time for any constant $c$.
\end{theorem}
\begin{proof} 
Let $k$ and $\ell$ denote the number of existential and universal variables in an input formula, respectively, and let $n = k + \ell$.
We consider two algorithms for \FEdQBF{3}.
The first one is the hypothetical FPT algorithm that solves \FEdQBF{3} in $\bigoh^*(c^k)$ time for constant $c$.
The second algorithm enumerates assignments to universal variables, and solves the resulting 3-SAT formulas.
The latter requires $\bigoh(2^\ell {c_3}^k)$ time,
where $c_3$ the infimum of all $b$ such that 3-SAT is solvable in $\bigoh^*(b^n)$ time.
Observe that $c_3 < 2$.
Now, let $\delta = k / n$ be the proportion of existential variables in the input formula.
We claim that, depending on the value of $\delta$, we can use
either the first or the second algorithm to solve \FEdQBF{3} in $\bigoh^*(d^n)$ time for some $d < 2$, contradicting SETH by Theorem~\ref{thm:SETH-for-AE3SAT}.
To this end, define $T = 1+\log_2(c)-\log_2(c_3)$, and observe that $T \geq 1$ since \FEdQBF{3} is more general than 3-SAT and $c \geq c_3$.

First, suppose $\delta > 1/T$, i.e. $k > n/T$.
Then use the first algorithm, which runs in $\bigoh^*(c^{n/T}) = \bigoh^*(2^{(\log_2(c)/T)n})$ time. 
Now suppose $\delta \leq 1/T$, i.e. $k \leq n / T$.
Then use the second algorithm, which runs in $\bigoh^*(2^{(1-\delta)n} {c_3}^{\delta n})$ time.
Observe that $(1 - \delta) + \log_2(c_3) \delta \leq 1 - (1 - \log_2(c_3)) / T \leq \log_2(c) / T$.
In both cases, our algorithm runs in $\bigoh^*(2^{(\log_2(c)/T) n})$ time.
Since $c_3 < 2$, we have $\log_2(c_3) < 1$ and $\log_2(c) < T$, which completes the proof.
\end{proof}

\section{Discussion} \label{sec:discussion}

In this paper we investigated a simple and overlooked parameter for QBFSAT and proved FPT with respect to the number of existentially quantified variables and the maximum arity of any clause. This parameterization is particularly noteworthy since applied QBF solving is frequently based on the idea of expanding universally quantified variables in order to get an instance that can be solved with SAT techniques. This strategy comes with the downside that (1) after removing universally quantified variables one still needs to solve an NP-hard problem, and (2) the strategy is inefficient for instances with many universally quantified variables. Our result is complementary in the sense that instances with many universal but few existential variables can now be handled efficiently with our novel FPT algorithm. While we in this paper concentrate on the theory it is natural to speculate whether these two approaches can be merged in actual QBF solvers to solve previously intractable instances faster.

For improvements, there is a gap between our $2^{\bigoh(k2^k)}$ FPT algorithm and our lower bound (under the SETH) which rules out any single-exponential $2^{\bigoh(k)}$ algorithm. It is not immediately which direction could be strengthened and new algorithmic ideas would likely be needed to bring down the running time to $2^{\bigoh(k^c)}$ for some fixed $c$. 
It would also be interesting to generalize our FPT algorithm to even broader classes of problems. A promising candidate is the {\em depedency quantified Boolean formula} (DQBF) formalism, i.e. Boolean formulas equipped with Henkin quantifiers. This problem is generally NEXPTIME-complete and has comparably few FPT results. Naturally, Lemma~\ref{lem:qbftocgis} would need to be modified to the DQBF setting, but besides that the main ideas should carry over.

From a purely theoretical perspective rather little is known about the logical fragment where we allow unrestricted universal quantification but only limited existential quantification. As a starting point one could define a closure operator on sets of Boolean relations induced by logical formulas allowing universal but no existential quantification over conjunctions of atoms from a predetermined structure. Such formulas would generalize {\em quantifier-free primitive positive} definitions (qfpp-definitions) which has been used to study fine-grained complexity aspects of CSPs~\cite{DBLP:journals/toct/LagerkvistW22}, but be more restrictive than the formulas considered by B\"orner et al.~\cite{borner2003} developed to study classical complexity of QCSPs. Could it, for example, be possible to give a classification akin to Post's classification of Boolean clones, or find a reasonable notion of algebraic invariance? A reasonable guess extending B\"orner et al.~\cite{borner2003} would be to consider algebras consisting of partial, surjective polymorphisms.

\section*{Acknowledgements}

The second author is partially supported by the Swedish research council under grant VR-2022-03214.
The fourth author was supported by
the Wallenberg AI, Autonomous Systems and Software Program (WASP) funded
by the Knut and Alice Wallenberg Foundation.

\bibliographystyle{abbrv}
\bibliography{references}

\end{document}